\documentclass[11pt]{article}
\usepackage[utf8]{inputenc}
\usepackage{fullpage}
\usepackage{amsmath, amsthm, amssymb, mathtools, xcolor, tikz, framed, float}
\usepackage{caption, comment}
\usepackage{enumerate}
\usetikzlibrary{patterns, patterns.meta, decorations.pathreplacing, calligraphy, shapes, intersections, positioning}
\usepackage[bottom]{footmisc}

\usepackage{algorithm}
\usepackage[noend]{algpseudocode}

\usepackage{pdflscape}

\usepackage{hyperref}
\hypersetup{
    colorlinks=true,
    linkcolor=blue,
    filecolor=magenta,
    citecolor = cyan, 
    urlcolor=cyan,
    }
    
\tikzstyle{comm}=[ellipse,draw=black,fill=gray!5]
\tikzstyle{leaf}=[]

\allowdisplaybreaks

\newcommand{\ket}[1]{|#1\rangle}

\newtheorem{theorem}{Theorem}[section]

\newtheorem{remark}[theorem]{Remark}
\newtheorem{lemma}[theorem]{Lemma}
\newtheorem{fact}[theorem]{Fact}
\newtheorem{claim}[theorem]{Claim}

\newtheorem{observation}[theorem]{Observation}
\newtheorem{definition}[theorem]{Definition}

\renewcommand{\epsilon}{\varepsilon}
\newcommand{\cbra}[1]{\left\{#1\right\}}
\newcommand{\rbra}[1]{\left(#1\right)}

\newcommand{\mathify}[1]{\ifmmode{#1}\else\mbox{$#1$}\fi}

\newcommand{\zone}{\{0, 1\}}
\newcommand{\zoneu}{\{0, 1, \mathsf{u}\}}

\newcommand{\res}{\mathsf{Res}}
\DeclarePairedDelimiter\abs{\lvert}{\rvert}%

\makeatletter
\let\oldabs\abs
\def\abs{\@ifstar{\oldabs}{\oldabs*}}
\makeatother

\newcommand{\OR}{\mathsf{OR}}

\newcommand{\IND}{\mathsf{IND}}
\newcommand{\mIND}{\mathsf{mIND}}

\newcommand{\A}{\mathcal{A}}

\newcommand{\su}{\mathsf{u}}
\renewcommand{\ss}{\mathsf{s}}

\newcommand{\bin}{\mathsf{bin}}
\newcommand{\bs}{\mathsf{bs}}
\newcommand{\ubs}{\mathsf{bs}_\su}
\newcommand{\ubszero}{\mathsf{bs}_{\su,0}}
\newcommand{\ubsone}{\mathsf{bs}_{\su,1}}
\newcommand{\us}{\mathsf{s}_\su}
\newcommand{\sC}{\mathsf{C}}
\newcommand{\usC}{\mathsf{C}_\su}
\newcommand{\usCb}{\mathsf{C}_{\su, b}}
\newcommand{\usCzero}{\mathsf{C}_{\su, 0}}
\newcommand{\usCone}{\mathsf{C}_{\su, 1}}

\newcommand{\uf}{f_\su}
\newcommand{\sD}{\mathsf{D}}
\newcommand{\sDu}{\mathsf{D}_\mathsf{u}}
\newcommand{\sR}{\mathsf{R}}
\newcommand{\sRu}{\mathsf{R}_\mathsf{u}}
\newcommand{\sQ}{\mathsf{Q}}
\newcommand{\sQu}{\mathsf{Q}_\mathsf{u}}
\newcommand{\uquery}{$\su$-query }

\title{Query Complexity with Unknowns}
\author{
Nikhil S.~Mande\thanks{University of Liverpool, UK. {\tt Nikhil.Mande@liverpool.ac.uk}}
\and 
Karteek Sreenivasaiah\thanks{University of Liverpool, UK. {\tt Karteek.Sreenivasaiah@liverpool.ac.uk}}
}
\date{}

\begin{document}

\maketitle

\begin{abstract}
    We initiate the study of a new model of query complexity of Boolean functions where, in addition to $0$ and $1$, the oracle can answer queries with ``unknown''. The query algorithm is expected to output the function value if it can be conclusively determined by the partial information gathered, and it must output ``unknown'' if not. We formalize this model by using Kleene's strong logic of indeterminacy on three variables to capture unknowns. We call this model the `$\su$-query model'.

    We relate the query complexity of functions in the new $\su$-query model with their analogs in the standard query model. We show an explicit function that is exponentially harder in the $\su$-query model than in the usual query model. We give sufficient conditions for a function to have $\su$-query complexity asymptotically the same as its query complexity.

    Using $\su$-query analogs of the combinatorial measures of sensitivity, block sensitivity, and certificate complexity, we show that deterministic, randomized, and quantum $\su$-query complexities of all total Boolean functions are polynomially related to each other, just as in the usual query models.
\end{abstract}

\section{Introduction}\label{sec: intro}

Query complexity of Boolean functions is a well-studied area of computer science. The model of computation studied involves a Boolean function $f$ (fixed beforehand), and query access to an oracle holding an input $x$. We are allowed to query the oracle for bits of $x$ and the goal is to compute $f(x)$. The complexity measure of primary interest is the number of queries made, in the worst case, to determine $f(x)$. 

In this paper, we initiate a systematic study of query complexity on a three-valued logic, namely Kleene's strong logic of indeterminacy  \cite{Kleene52}, usually denoted \emph{K3}. The logic K3 has three truth values $0$, $1$, and $\su$. The behaviour of the third value $\su$ with respect to the basic Boolean primitives --- conjunction ($\wedge$), disjunction ($\vee$), and negation ($\neg$) --- are given in Table~\ref{table:K3}.

\begin{table}[]
\centering
\begin{tabular}{l|lll}
$\wedge$ & 0 & \textsf{u} & 1 \\
\hline
0 & 0 & 0 & 0 \\
\textsf{u} & 0 & \textsf{u} & \textsf{u} \\
1 & 0 & \textsf{u} & 1
\end{tabular}
\quad
\begin{tabular}{l|lll}
$\vee$  & 0 & \textsf{u} & 1 \\
\hline
0 & 0 & \textsf{u} & 1 \\
\textsf{u} & \textsf{u} & \textsf{u} & 1 \\
1 & 1 & 1 & 1
\end{tabular}
\quad
\begin{tabular}{l|lll}
$\neg$ & 0 & \textsf{u} & 1 \\
\hline
    & 1 & \textsf{u} & 0
\end{tabular}

\caption{Kleene's three valued logic `K3'}
\label{table:K3}
\end{table}

The logic K3 is typically used in models where there is a need to represent, and work with, \emph{unknowns} and hence has found several important wide-ranging applications in computer science. For instance, in relational database theory, SQL implements K3 and uses $\su$ to represent a NULL value \cite{Meyden98}. Perhaps the oldest use of K3 is in modeling \emph{hazards} that occur in real-world combinational circuits. Recently there have been a series of results studying constructions and complexity of hazard-free circuits~\cite{IKLLMS19, Jukna21, IkenmeyerKS23}. Here $\su$ was used to represent an \emph{unstable} voltage that can resolve to $0$ or $1$ at a later time. 

One way to interpret the basic binary operations $\vee$, $\wedge$, and $\neg$ in K3 is as follows: for a bit $b \in \zone$, if the value of the function is guaranteed to be $b$ regardless of all $\zone$-settings to the $\su$-variables, then the output is $b$. Otherwise, the output is $\su$. This interpretation generalizes in a nice way to $n$-variate Boolean functions: given $f : \zone^n \to \zone$, an input to $f$ in K3 would be a string in $\zoneu^n$, and the output is $b \in \zone$ if and only if $f$ evaluates to $b$ for all settings of the $\su$-variables in $\zone$. Otherwise, the output is $\su$. We denote this extension $\uf$ and formally define it below. In literature, this extension is typically called the \emph{hazard-free} extension of $f$ (see, for instance,~\cite[Equation~1]{IkenmeyerKS23}), and is an important concept studied in circuits and switching network theory since the 1950s. The interested reader can refer to \cite{IKLLMS19}, and the references therein, for history and applications of this particular way of extending $f$ to K3.

Formally, for an input $x \in \zoneu^n$, define the \emph{resolutions} of $x$ as follows:
\[
    \res(x) := \cbra{y \in \zone^n : y_i = x_i~\forall i \in [n]~\textnormal{with}~x_i \in \zone}.
\]
That is, $\res(x)$ denotes the set of all strings in $\zone^n$ that are consistent with the $\{0,1\}$-valued bits of $x$.
For a Boolean function $f : \zone^n \to \zone$, define its extension to K3, denoted $\uf : \zoneu^n \to \zoneu$, by
\[
    \uf(y) = \begin{cases}
    b \in \zone & f(x) = b~\forall x \in \res(y)\\
    \su & \textnormal{otherwise}.
    \end{cases}
\]

We initiate a study of the \emph{query complexity} of these extensions of Boolean functions. Here, in addition to $0$ and $1$, the oracle can answer a query with ``unknown'', which is denoted $\su$. This can model scenarios where the oracle itself holds incomplete information, or would like to withhold some information from the querying agent. 

The deterministic $\su$-query complexity of $f$, denoted $\sDu(f)$, is defined as maximum over all inputs of length $n$, the number of queries by a best query algorithm for $\uf$.

While the query complexity of $n$-variate Boolean functions that depend on all input variables is well known to be between $\log n$ and $n$, the analogous bounds are no longer clear in the $\su$-query setting. In particular, the Indexing function which admits a simple $O(\log n)$ query algorithm, can be seen to have linear $\su$-query complexity. In light of this exponential gap, it is natural to ask --- do all functions require $n^{\Omega(1)}$ $\su$-query complexity? We show that this is not true, and exhibit an explicit example of a non-degenerate Boolean function that has \uquery complexity $O(\log n)$. 

Towards the more general goal of understanding what properties of Boolean functions play a role in determining their \uquery complexity, we give sufficient conditions under which \uquery complexity is asymptotically the same as standard query complexity. Further, analogous to existing measures, we redefine notions of sensitivity, block sensitivity, and certificate complexity for functions over the logic K3. We show that $\sDu(f)$ is related to these measures in a similar fashion as is known in query complexity.

\subsection{Our Results}
We relate deterministic \uquery complexity of Boolean functions to their query complexity in Section~\ref{sec: relating query and uquery}. It is well known that the Indexing function on $n + 2^n$ input variables (see Definition~\ref{defn: indexing}) has $O(n)$ query complexity. In contrast, we show that its $\su$-query complexity is polynomially large in the number of variables.
\begin{lemma}\label{lem: unstable indexing lower bounds}
    Let $n$ be a positive integer. Then,
    \[
        \sDu(\IND_n) = \Theta(2^n), \qquad \sRu(\IND_n) = \Theta(2^n), \qquad \sQu(\IND_n) = \Theta(2^{n/2}).
    \]
\end{lemma}
We prove this lemma by showing a reduction from the OR function on $2^n$ variables, which is well known to have $\Theta(2^n)$ deterministic and randomized query complexities, and $\Theta(2^{n/2})$ quantum query complexity. 
Given that the Indexing function is considered to be one of the textbook ``easy'' functions for query complexity, this raises the question of whether the $\su$-query complexity of all Boolean functions is polynomially large. We show that this is not true, witnessed by a monotone variant of the Indexing function (see Definition~\ref{defn: monotone indexing}). We show that this function admits a $\su$-query protocol of logarithmic cost in its input size (see Lemma~\ref{lem: monotone indexing upper bound}). More generally, we show that for monotone functions, their query complexity and \uquery complexity are asymptotically the same.

\begin{lemma}\label{lem: monotone unstable}
    Let $f : \zone^n \to \zone$ be a monotone Boolean function. Then we have
    \[
        \sDu(f) = \Theta(\sD(f)), \qquad \sRu(f) = \Theta(\sR(f)), \qquad \sQu(f) = \Theta(\sQ(f)).  
    \]
\end{lemma}

In Section~\ref{sec: poly relationships}, we show that the deterministic, randomized, and quantum \uquery complexities of a function are all polynomially related to each other.
\begin{theorem}\label{thm: main}
    For $f : \zone^n \to \zone$, we have
    \[
    \sDu(f) = O(\sRu(f)^3), \qquad \sDu(f) = O(\sQu(f)^6).
    \]
\end{theorem}
An analogous statement has been long known to hold true in the usual query models~\cite{Nisan91, NS94, BBC+01}. All of these proofs proceed by showing polynomial relationships of these query measures with the intermediate combinatorial measure of \emph{block sensitivity}. While we are able to adapt these known proofs to yield a proof of Theorem~\ref{thm: main}, we do have to redefine the relevant combinatorial measures appropriately and suitably modify certain steps to ensure that the analogous results in the \uquery model still hold. 

\section{Preliminaries}
In this section we review required preliminaries. In particular, we define the \uquery model, and the analogs of combinatorial measures that are usually defined for Boolean functions such as block sensitivity and certificate complexity.

All logarithms in this paper are taken base 2. Unless mentioned otherwise, all Boolean functions $f : \zone^n \to \zone$ considered in this paper are non-degenerate, that is, they depend on all input variables. Formally, this means for all $i \in [n]$, there exists a string $x^i \in \zone^n$ such that $f(x^i) \neq f(x^i \oplus e_i)$, where $e_i$ denotes the $n$-bit string that has a single 1 in the $i$'th location, and $\oplus$ represents bitwise XOR. For a string $x \in \zone^n$, let $|x|$ denote $|\cbra{i \in [n] : x_i = 1}|$. For a string $x \in D^n$ and a set $S \subseteq [n]$, let $x_S$ denote the string $x$ restricted to the indices in $S$. Further, for a string $w\in D^n$, we use $x|_{S\leftarrow w}$ to denote the string obtained from $x$ by replacing the values in positions in $S$ with that of $w$. More formally, if $y = x|_{S\leftarrow w}$, then $y_S=w$ and $y_{[n]\setminus S}=x_{[n]\setminus S}$.
For two strings $x,y\in \zone^n$, we say $x\le y$ if for all $i\in [n]$, we have $x_i\le y_i$.

\begin{definition}
    A Boolean function $f:\zone^n\to \zone$ is \emph{monotone} if for all $x\le y$, we have $f(x)\le f(y)$.
\end{definition}
A generalization of monotone Boolean functions is \emph{unate} functions:
\begin{definition}
\label{defn: unate}
    A Boolean function $f:\zone^n\to \zone$ is \emph{unate} if there exists a monotone Boolean function $g:\zone^n\to \zone$, and a string $s\in \zone^n$, such that for all inputs $x\in \zone^n$, we have $f(x) = g(x\oplus s)$. The string $s$ is often called the \emph{orientation} of $f$.
\end{definition}

Define the Indexing function as follows.
\begin{definition}\label{defn: indexing}
    For an integer $n > 0$, and a string $x \in \zone^n$, let $\bin(x)$ denote the integer in $[2^n]$ whose binary representation is given by $x$. Define the Indexing function on $n + 2^n$ variables, denoted $\IND_{n} : \zone^{n + 2^n} \to \zone$ by
    \[
        \IND_n(x, y) = y_{\bin(x)}.
    \]
    We refer to the $x$-variables as \emph{addressing variables} and the $y$-variables as \emph{target variables}.
\end{definition}
The monotone variant of the Indexing function from \cite{Wegener85} is defined as follows.
\begin{definition}[\cite{Wegener85}]\label{defn: monotone indexing}
For an even integer $n > 0$, define $\mIND_n$, as follows: The function $\mIND_n$ is defined on $n + \binom{n}{n/2}$ (which is $\Theta(2^n/\sqrt{n})$) variables, where the latter $\binom{n}{n/2}$ variables are indexed by all $n$-bit strings of Hamming weight exactly $n/2$. For $(x, y) \in \zone^{n + \binom{n}{n/2}}$, define
\[
    \mIND_n(x, y) = \begin{cases}
    0 & |x| < n/2\\
    1 & |x| > n/2\\
    y_x & \textnormal{otherwise}.
    \end{cases}
\]    
\end{definition}

We use Yao's minimax principle~\cite{Yao77}, stated below in a form convenient for us.
\begin{lemma}[Yao's minimax principle]\label{lem: yao}
    For finite sets $D, E$ and a function $f: D \to E$, we have $\sR(f) \geq k$ if and only if there exists a distribution $\mu : D \to [0,1]$ such that $\sD_\mu(f) \geq k$.
    Here, $\sD_\mu(f)$ is the minimum depth of a deterministic decision tree that computes $f$ to error at most $1/3$ when inputs are drawn from the distribution $\mu$.
\end{lemma}

We require the adversary method for lower bounds on quantum query complexity due to Ambainis~\cite{Amb02}. Ambainis' bound is only stated for functions with Boolean input and Boolean output, but the same results extend to functions with non-Boolean input and non-Boolean outputs as well, see~\cite[Theorem~3.1]{SS06}.
\begin{theorem}[{\cite[Theorem~5.1]{Amb02}}]\label{thm: ambainis}
    Let $D, E$ be finite sets, let $n$ be a positive integer, and let $f : D^n \to E$ be a function. Let $X, Y$ be two sets of inputs such that $f(x) \neq f(y)$ for all $(x, y) \in X \times Y$. Let $R \subseteq X \times Y$ be such that
    \begin{itemize}
        \item For every $x \in X$, there are at least $m$ different $y \in Y$ with $(x, y) \in R$,
        \item for every $y \in Y$, there are at least $m'$ different $x \in X$ with $(x, y) \in R$,
        \item for every $x \in X$ and $i \in [n]$, there are at most $\ell$ different $y \in Y$ such that $(x, y) \in R$ and $x_i \neq y_i$,
        \item for every $y \in Y$ and $i \in [n]$, there are at most $\ell'$ different $x \in X$ such that $(x, y) \in R$ and $x_i \neq y_i$.
    \end{itemize}
    Then $\sQ(f) = \Omega\rbra{\frac{mm'}{\ell\ell'}}$.
\end{theorem}

\subsection{The $\su$-query model}
Recall from Section~\ref{sec: intro} that for a string $x \in \zoneu^n$,
\[
    \res(x) := \cbra{y \in \zone^n : y_i = x_i~\forall i \in [n]~\textnormal{with}~x_i \in \zone}.
\]
$\res(x)$ denotes the set of all strings in $\zone^n$ that are consistent with the $\zone$-valued bits of $x$.
Also recall that for a Boolean function $f : \zone^n \to \zone$, its K3 extension, denoted $\uf : \zoneu^n \to \zoneu$, is defined by
\[
    \uf(y) = \begin{cases}
    b \in \zone & f(x) = b~\forall x \in \res(y)\\
    \su & \textnormal{otherwise}.
    \end{cases}
\]
That is, $\uf(y) = b \in \zone$ if $\uf$ evaluates to the constant $b$ in the subcube of $\zone^n$ defined by all inputs consistent with the $\zone$-valued bits of $y$, and $\uf(y) = \su$ otherwise.

For the discussion below, fix a positive integer $n$ and a Boolean function $f : \zone^n \to \zone$.

A \emph{deterministic decision tree}, also called a \emph{deterministic query algorithm}, computing $\uf$, is a ternary tree $T$ whose leaf nodes are labeled by elements of $\zoneu$, each internal node is labeled by an index $i \in [n]$ and has three outgoing edges, labeled $0, 1$, and $u$.
On an input $x \in \zoneu^n$, the tree's computation proceeds from the root down to a leaf as follows: compute $x_i$ as indicated by the node's label and follow the edge indicated by the value of $x_i$. Continue this way until reaching a leaf, at which point the value of the leaf is output.
The correctness requirement of the algorithm is that the output of $T$ on input $x \in \zoneu^n$ equals $\uf(x)$ for all $x \in \zoneu^n$. In this case, $T$ is said to compute $\uf$.
The cost of a query algorithm is the number of queries made on a worst-case input, which is exactly the depth of the corresponding tree.
Formally, the deterministic \uquery complexity complexity of $f$, denoted $\sDu(f)$, is defined as
\[
    \sDu(f) := \min_{T} \textnormal{depth}(T),
\]
where the minimization is over all deterministic decision trees $T$ that compute $\uf$.
A randomized decision tree is a distribution over deterministic decision trees. We say a randomized decision tree computes $\uf$ with error $1/3$ if for all $x \in \zoneu^n$, the probability that it outputs $\uf(x)$ is at least $2/3$. The depth of a randomized decision tree is the maximum depth of a deterministic decision tree in its support.
Define the randomized \uquery complexity of $f$ as follows.
\[
    \sRu(f) := \min_{T}\textnormal{depth}(T),
\]
where the minimization is over all randomized decision trees $T$ that compute $\uf$ to error at most $1/3$.
A quantum query algorithm $\mathcal{A}$ for $\uf$ begins in a fixed initial state $\ket{\psi_0}$ in a finite-dimensional Hilbert space, applies a sequence of unitaries $U_0, O_x, U_1, O_x, \dots, U_T$, and performs a measurement. Here, the initial state $\ket{\psi_0}$ and the unitaries $U_0, U_1, \dots, U_T$ are independent of the input. The unitary $O_x$ represents the ``query'' operation, and does the following for each basis state:
it maps $\ket{i}\ket{b}\ket{w}$ to $\ket{i}\ket{b + x_i \mod 3}\ket{w}$ for all $i \in [n]$ (here $x_i = \su$ is interpreted as $x_i = 2$, and the last register represents workspace that is not affected by the application of a query oracle).
The algorithm then performs a 3-outcome measurement on a designated output qutrit and outputs the observed value.
We say that $\mathcal{A}$ is a bounded-error quantum query algorithm computing $\uf$ if for all $x \in \zoneu^n$ the probability that $\uf(x)$ is output is at least $2/3$. The (bounded-error) \emph{quantum query complexity of $\uf$}, denoted by $\sQu(f)$, is the least number of queries required for a quantum query algorithm to compute $\uf$ with error probability at most $1/3$.

\begin{definition}[Sensitive block, block sensitivity]\label{defn: bs}
    Let $f : \zone^n \to \zone$.
    A set $B \subseteq [n]$ is a \emph{sensitive block of $\uf$ at input $x$} if there exists $y \in \zoneu^n$ such that $y_{[n] \setminus B} = x_{[n]\setminus B}$ and $\uf(y) \neq \uf(x)$.
    A sensitive block of $\uf$ at $x$ is \emph{minimal} if no proper subset of it is a sensitive block of $\uf$ at $x$.
    The block sensitivity of $\uf$ at $x$, denoted by $\bs(\uf,x)$, is the maximum $k$ for which there exist disjoint blocks, $B_1, B_2, \ldots, B_k$, each of which is a
    sensitive block of $\uf$ at $x$. The \emph{$\su$-block sensitivity of $f$}, denoted by $\ubs(f)$, is defined as
    \[
        \ubs(f)=\max_{x\in \zoneu^n} \bs(\uf,x).
    \]
\end{definition}

For a string $x \in \zoneu^n$ and a set $B \subseteq [n]$ that is a sensitive block of $\uf$ at $x$, we abuse notation and use $x^B$ to denote an arbitrary but fixed string $y \in \zoneu^n$ that satisfies $y_{[n] \setminus B} = x_{[n] \setminus B}$ and $\uf(y) \neq \uf(x)$.
An index $i\in[n]$ is sensitive for $x$ with respect to $\uf$ if $\{i\}$ is a sensitive block of $\uf$ at $x$. If we only restrict attention to size-1 blocks in Definition~\ref{defn: bs}, we call the corresponding measure \emph{$\su$-sensitivity of $f$} and denote it by $\us(f)$.

By a \emph{partial assignment} on $n$ bits, we mean a string $p \in \cbra{0, 1, \su, *}^n$, representing partial knowledge of a string in $\zoneu^n$, where the $*$-entries are yet to be determined. We say a string $y \in \zoneu^n$ is \emph{consistent} with a partial assignment $p \in \cbra{0, 1, \su, *}^n$ if $y_i = p_i$ for all $i \in [n]$ with $p_i \neq *$.

\begin{definition}[Certificate]
    Let $f : \zone^n \to \zone$ and $x \in \zoneu$. A partial assignment $p \in \cbra{0, 1, \su, *}^n$ is called a \emph{certificate} for $\uf$ at $x$ if 
    \begin{itemize}
        \item $x$ is consistent with $p$, and
        \item $\uf(y) = \uf(x)$ for all $y$ consistent with $p$.
    \end{itemize}
    The size of this certificate is $|p| := \abs{\cbra{i \in [n] : p_i \neq *}}$. The \emph{domain of $p$} is said to be $\cbra{i \in [n] : p_i \neq *}$.
\end{definition}
In other words, a certificate for $\uf$ at $x$ is a set of variables of $x$ that if revealed, guarantees the output of all consistent strings with the revealed variables to be equal to $\uf(x)$.

The \emph{certificate complexity of $\uf$ at $x \in \zoneu^n$}, denoted $\sC(\uf,x)$, is the minimum size of a certificate $p$ for $\uf$ at $x$.

\begin{definition}[Certificate complexity]\label{defn: ucert}
    For $b \in \zoneu$, define the $b$-certificate complexity of $f$, denoted $\usCb(f)$, to be the maximum value of $\sC(\uf, x)$ over all $x \in \uf^{-1}(b)$. For a function $f : \zone^n \to \zone$, the \emph{certificate complexity} of $f$, denoted $\usC(f)$, is the maximum value of $\usCzero(\uf)$ and $\usCone(f)$.
\end{definition}

Observe that we did not include inputs $x \in \uf^{-1}(u)$ in the maximization in the definition of $\usC(f)$ in Definition~\ref{defn: ucert}. We justify below that doing so does not affect $\usC(f)$ asymptotically.
\begin{observation}
    Let $f : \zone^n \to \zone$ and let $x \in \uf^{-1}(\su)$. Then,
    \[
        \mathsf{C}_{\su, \su}(f) \leq 2 \usC(f).
    \]
\end{observation}
\begin{proof}
    Let $x \in \uf^{-1}(\su)$. By definition, this means there exist strings $y, z \in \zone^n$ such that $y \in \res(x) \cap \uf^{-1}(1)$ and $z \in \res(x) \cap \uf^{-1}(0)$. Let $c_y, c_z$ be a 0-certificate for $\uf$ at $y$ and a 1-certificate for $\uf$ at $z$, respectively. Let their domains be $C_y, C_z$, respectively. Let $p \in \cbra{0, 1, \su, *}^n$ be the partial assignment defined by
    \[
        p_i = \begin{cases}
            \su & i \in C_y \cup C_z,\\
            * & \textnormal{otherwise}.
        \end{cases}
    \]
    Note that there is a setting of the $\su$-variables in $p$ (set the variables in $C_y$ consistently with $c_y$, and leave the other variables as $\su$) that forces $\uf$ to 0, and another setting to the $\su$-variables in $p$ (set the variables in $C_z$ consistently with $c_z$, and leave the other variables as $\su$) that forces $\uf$ to 1. Thus, $p$ certifies that $\uf(x) = \su$, and the number of non-$*$ variables in it is at most $2\usC(f)$.
\end{proof}

We require the following observations.
\begin{observation}\label{obs: unstable certificate no u}
    Let $f : \zone^n \to \zone$, $b \in \zone$ and $x \in \uf^{-1}(b)$. Let $c$ be a minimal certificate of $\uf$ at $x$, and let its domain be $C$. Then for all $i\in C$, we have $c_i \neq \su$.
\end{observation}
\begin{proof}
Assume towards a contradiction an index $i \in [n]$ in a minimal certificate $c$ for $\uf$ at $x$ with $c_i = \su$. By the definition of a certificate, and $\uf$, the partial assignment $c'$ obtained by removing $i$ from the domain of $c$ is such that all strings $x \in \zoneu^n$ consistent with $c'$ satisfy $\uf(x) = b$. This contradicts the minimality of $c$.
\end{proof}
\begin{observation}\label{obs: unstable zero one certificates intersect}
    Let $f : \zone^n \to \zone, x \in \uf^{-1}(0)$ and $y \in \uf^{-1}(1)$. Let $c^x, c^y$ be minimum certificates for $\uf$ at $x$ and $y$, respectively, and let $C^x$ and $C^y$ be their domains, respectively. Then there exists an index $i \in C^x \cap C^y$ such that one of the following holds:
    \begin{itemize}
        \item $c^x_i = 0$ and $c^y_i = 1$, or
        \item $c^x_i = 1$ and $c^y_i = 0$.
    \end{itemize} 
\end{observation}
\begin{proof}
    First recall from Observation~\ref{obs: unstable certificate no u} that neither $c^x$ nor $c^y$ take the value $\su$ on their domains. Next, towards a contradiction, assume there is no index $i \in C^x \cap C^y$ satisfying the requirements in the statement of the observation. This means $x \in \uf^{-1}(0)$ is a valid completion of $c^y$ and $y \in \uf^{-1}(1)$ is a valid completion of $c^x$. This contradicts the assumption that $c^x$ and $c^y$ were certificates.
\end{proof}

\section{Relating query complexity and \uquery complexity}\label{sec: relating query and uquery}
In this section, we examine the relationship between the well-studied notion of query complexity and \uquery complexity. We begin by showing that the Indexing function exhibits an exponential separation between usual query complexity and \uquery complexity.

The following is well known.
\begin{fact}
    Let $n$ be a positive integer. Then,
    \[
        \sD(\IND_n) = n+1, \qquad \sR(\IND_n) = \Omega(n), \qquad \sQ(\IND_n) = \Omega(n).
    \]
\end{fact}

Define the OR function as
\[
    \OR_n(x) = \begin{cases}
    0 & x = 0^n\\
    1 & \textnormal{otherwise}.
    \end{cases}
\]
\begin{lemma}\label{lem: pror lower bounds}
    Let $n$ be a positive integer. Then,
    \[
        \sD(\OR_n) = n, \qquad \sR(\OR_n) = \Theta(n), \qquad \sQ(\OR_n) = \Theta(\sqrt{n}).
    \]
\end{lemma}
The deterministic and randomized bounds are folklore, the quantum upper bound follows by Grover's search algorithm~\cite{Gro96}, and the quantum lower bound was first shown in~\cite{BBBV97}.

We now prove Lemma~\ref{lem: unstable indexing lower bounds} using a reduction from the OR function. 
\begin{proof}[Proof of Lemma~\ref{lem: unstable indexing lower bounds}]
    Let $f = \IND_n$. We first show the upper bounds. The deterministic and randomized upper bounds are trivial since the number of input bits is $\Theta(2^n)$. For the quantum upper bound, we first query the $n$ addressing variables, let the output be $x \in \zoneu^n$. Further, let $z \in \zoneu^{2^n}$ denote the input's target variables. Consider the set of indices $I = \cbra{\bin(y) \in [2^n] : y \in \res(x)}$. By the definition of $\IND_n$, we have $\uf(x, z) = b \in \zone$ iff $z_i = b$ for all $i \in I$. By a standard use of Grover's search algorithm a constant number of times, we can determine with high probability whether $z$ is the constant 0 on $I$, or whether it is the constant 1 on $I$, or whether neither of these is the case. This immediately gives us the output: 0 in the first case, 1 in the second case, and $\su$ in the last case. The cost of querying the addressing variables is $n$, and constantly many applications of Grover's search algorithm costs $O(2^{n/2})$.
    
    We now show the lower bounds. Given a (deterministic/randomized/quantum) algorithm $\A$ solving $\uf$, we construct a (deterministic/randomized/quantum) algorithm $\A'$ with query complexity at most that of $\A$, and that solves $\OR_{2^n}$. Let the input to $\OR_{2^n}$ be denoted by $x$.

    We define $\A'$ to be the algorithm that simulates a run of $\A$ in the following way:
    
    \begin{enumerate}
        \setlength\itemsep{-.3em}
        \item If $\A$ queries the $j$'th bit of its input where $j \in [n]$, return $\su$.
        \item If $\A$ queries the $j$'th bit of its input, where $j = n + k$ for some $k > 0$, then query $x_k$ and return that value.
    \end{enumerate}
    If $\A$ outputs $0$, $\A'$ outputs $0$. Else $\A'$ outputs $1$. This can be simulated in the deterministic, randomized, and quantum models.
    Note that the query complexity of $\A'$ is at most that of $\A$. We now argue correctness.
    
    \begin{itemize}
        \item Let $x = 0^{2^n}$. The input $z \in \zoneu^{n + 2^n}$ to $\uf$ with all addressing bits set to $\su$, and the target bits all equal to $0$ is consistent with the run of $\A$ above. Note that for all $y \in \res(z)$ we have $f(y) = 0$ since the relevant target bit is always 0 regardless of the choice of addressing bits. 
        Since (with high probability) $\A$ must be correct on $z$, this means that the output of $\A$, and hence $\A'$, is 0 (with high probability) and correct in this case.
        \item Let $x \neq 0^{2^n}$. The input $z \in \zoneu^{n + 2^n}$ to $\uf$ with all addressing bits set to $\su$, and the target bits equal to the string $x$ is consistent with the run of $\A$ above. Now there exists a $y \in \res(z)$ that sets all indexing variables to point to a non-0 bit of the target variables (since there exists such a target variable). By the correctness of $\A$ (with high probability), this means the output of $\A$ must either be $\su$ or 1. Thus $\A'$ outputs 1 in this case (with high probability), which is the correct answer.
    \end{itemize}
\end{proof}

Now that we have exhibited an explicit function in Lemma \ref{lem: unstable indexing lower bounds} that witnesses an exponential gap between each of deterministic, randomized, and quantum query complexity and their $\su$ counterparts, we turn to the question of what functions are as hard in the standard query model as they are in the \uquery model. We now prove Lemma~\ref{lem: monotone unstable}, which says that that the class of monotone functions is one such class of functions.

\begin{proof}[Proof of Lemma~\ref{lem: monotone unstable}]
    Any query algorithm for $\uf$ also computes $f$ with at most the same cost, and hence $\sD(f) \le \sDu(f)$. Similarly we have $\sR(f) \leq \sRu(f)$ and $\sQ(f) \leq \sQu(f)$.

    We start with a best (deterministic/randomized/quantum) query algorithm $\A$ for $f$, and an oracle holding an input $x \in \zoneu^n$ for $\uf$ .
    Now, for $b \in \zone$, we define $\A_b$ to be the same algorithm as $\mathcal{A}$, but whenever $\mathcal{A}$ queries the $j$'th bit of its input, it performs the following operation instead:
    
    \begin{enumerate}
        \setlength\itemsep{-.3em}
        \item Query the $j$'th bit of $x$, denote the outcome by $x_j$.
        \item If $x_j \in \{0,1\}$, return $x_j$.
        \item If $x_j = \su$, return $b$.
    \end{enumerate}

    Note that this operation can indeed be implemented quantumly making $2$ queries to $O_x$. The initial query performs the instructions described above, and the second query uncomputes the values from the tuple we don't need for the remaining part of the computation. Note here that $O_x^3 = I$ by definition, and thus $O_x^2 = O_x^{-1}$, which is what we need to implement for the uncomputation operations.
    
    Let $S_{\su} = \{i\mid x_i=\su \}$ be the positions that have $\su$ in $x$. Recall that $y^0 := x|_{S_u\leftarrow \vec{0}}$ (and $y^1 := x|_{S_u\leftarrow \vec{1}}$) is the input that has all the $\su$ in $x$ replaced by $0$s (by $1$s respectively). Run $\A_0$ and $\A_1$ (possibly repeated constantly many times each) to determine the values of $f(y^0)$ and $f(y^1)$ with high probability. If $f(y^0)=1$, then we output $1$. Else if $f(y^1)=0$, we output $0$. Else we have $f(y^0)=0$ and $f(y^1)=1$, and we output $\su$.
    
    \textbf{Correctness:} First observe that for $b \in \zone$, all answers to queries in the algorithm $\A_b$ are consistent with the input $y^b$. Next observe that in the poset (equivalently `subcube') formed by the resolutions of $x$, the inputs $y^0$ and $y^1$ form the bottom and top elements respectively. Since $f$ is monotone, if $f(y^0)=1$, we can conclude that $f$ is $1$ on all resolutions of $x$. Similarly when $f(y^1)=0$, it must be the case that $f$ is $0$ on all inputs in the poset. The remaining case is when $f(y^0)=0$ and $f(y^1)=1$. In this case, the inputs $y^0$ and $y^1$ are themselves resolutions of $x$ with different evaluations of $f$, and hence the algorithm correctly outputs $\su$. 

    By standard boosting of success probability using a Chernoff bound by repeating $\A_0$ ($\A_1$) constantly many times and taking a majority vote of the answers, we can ensure that the correctness probability of $\A_0$ ($\A_1$) is at least 0.9. Thus, the algorithm described above has correctness probability at least $0.9^2 = 0.81$, and its cost is at most a constant times the cost of $\A$.
\end{proof}

Specifically, Lemma~\ref{lem: monotone unstable} allows us to exhibit an explicit $n$-variate Boolean function $f$ with even \uquery complexity $O(\log n)$. The function of interest here is the monotone variant of the Indexing function due to Wegener~\cite{Wegener85}, defined in Definition~\ref{defn: monotone indexing}.

\begin{lemma}\label{lem: monotone indexing upper bound}
    Let $n$ be a positive integer and let $m = n + \binom{n}{n/2}$. Let $f = \mIND : \zone^{m} \to \zone$. Then,
    \[
        \sDu(f), \sRu(f), \sQu(f) = O(\log m).
    \]
\end{lemma}
\begin{proof}
    A deterministic (and hence randomized and quantum) query upper bound for $f$ is as follows: first query the $n$ indexing variables. If the number of 1's seen so far is smaller than $n/2$, output 0, else if the number of 1's seen so far is larger than $n/2$, output 1, else query the relevant target variable and output the value seen. This gives a deterministic query upper bound of $n = O(\log m)$. Since $f$ is monotone, using Lemma~\ref{lem: monotone unstable} completes the proof.
\end{proof}
This shows a stark contrast between the (deterministic/randomized/quantum) $\su$-query complexity of the monotone and non-monotone variants of the Indexing function, with the former having logarithmic (deterministic/randomized/quantum) $\su$-query complexity and the latter having linear (square root in the case of quantum) query complexity (Lemma~\ref{lem: unstable indexing lower bounds}).

\begin{observation}
    The proof of Lemma \ref{lem: monotone unstable} goes through for unate functions (Definition \ref{defn: unate}) too. Let $f$ be unate with orientation $s\in \zone^n$. Then, we modify the proof of Lemma \ref{lem: monotone unstable} to  use $s$ instead of $\vec{0}$ and $\overline{s}=\vec{1}\oplus s$ instead of $\vec{1}$. i.e., we want to compute $f$ on $y^0 = x|_{S_u\leftarrow s}$ and $y^1 = x|_{S_u\leftarrow \overline{s}}$. So if $x_j$ was queried and the reply was $\su$, we use $s_j$ ($\overline{s}_j$) when computing $f(y^0)$ ($f(y^1)$ respectively). It is not hard to see that the proof goes through exactly as before.
\end{observation}

For non-monotone Boolean functions, we can relate their \uquery complexity to the query complexity of their downward closure defined as follows:
\begin{definition}
Let $f:\{0,1\}^n\to \{0,1\}$. We define the downward closure $f^\nabla:\{0,1\}^n\to \{0,1\}$ of $f$ to be the following monotone Boolean function:
    \[
        f^\nabla(x) = \bigvee_{z\le x} f(z)
    \]
\end{definition}
Observe that $f^\nabla$ is monotone for any $f$, and coincides with $f$ when $f$ is monotone.
\begin{lemma}\label{lem: downward closure unstable}
    For all $f : \zone^n \to \zone$,
    \[
    \sD(f^\nabla) \leq \sDu(f), \qquad \sR(f^\nabla) \leq \sRu(f), \qquad \sQ(f^\nabla) \leq \sQu(f).
    \]
\end{lemma}
\begin{proof}
    Let $\A$ be a query algorithm for $\uf$. We can use $\A$ to compute $f^\nabla$ on an input $x\in \{0,1\}^n$ as follows. 
We define $\A_\su$ to be the same algorithm as $\mathcal{A}$, but whenever $\mathcal{A}$ queries the $j$'th bit of its input, it performs the following operation instead:
    
    \begin{enumerate}
        \setlength\itemsep{-.3em}
        \item Query the $j$'th bit of $x$, denote the outcome by $x_j$.
        \item If $x_j = 0$, return $x_j$.
        \item If $x_j = 1$, return $\su$.
    \end{enumerate}
If $\A_\su$ outputs $0$, we output $0$. Else we output $1$. As in the proof of Lemma~\ref{lem: monotone unstable}, this can be simulated in the deterministic, randomized, and quantum models.

    \textbf{Correctness:} If $\A_\su$ outputs $0$ on $y$, then it must be that every resolution of $y$ evaluates to $0$. Observe that, by definition of $y$, the set $\res(y)$ is precisely the set of strings $z$ such that $z\le x$. This means $f^\nabla(x) = \bigvee_{z\le x} f(z) = 0$. On the other hand, if $\A_\su$ outputs $1$ or $\su$, then in both cases there must exist an input $z\in \res(y)$ with $f(z)=1$. Hence $f^\nabla(x)=1$.
\end{proof}

\begin{remark}
    An analogous statement relating the monotone circuit complexity of $f^\nabla$ and hazard-free circuit complexity of $f$ was recently shown by Jukna \cite{Jukna21}.
\end{remark}

\section{Polynomial relationships between deterministic, randomized, and quantum \uquery complexities}\label{sec: poly relationships}
It was shown in a line of work~\cite{Nisan91, NS94, BBC+01} that for all Boolean functions $f : \zone^n \to \zone$, its deterministic, randomized, and quantum query complexities are polynomially related to each other as follows:
The polynomial equivalences can be shown using the following three inequalities:
\begin{enumerate}[(i)]
    \item $\sD(f) \leq \sC(f) \cdot \bs(f)$.
    \item $\sC(f) \leq \bs(f) \cdot \ss(f) \leq \bs(f)^2$.
    \item $\bs(f) = O(\sR(f))$, $\sqrt{\bs(f)} = O(\sQ(f))$.
\end{enumerate}

In this section we show that we can suitably adapt these proofs to also show that the \uquery complexity counterparts are also all polynomially related to each other. While the proofs are similar to the proofs of the above, we do have to take some care. We start by showing $\sDu(f) = O(\usC(f) \cdot \ubs(f))$. Algorithm \ref{algo:UpperBound} is a deterministic query algorithm that achieves this bound, as shown in Theorem \ref{thm: algo upper bound}. Following this, we show \uquery versions of (ii) and (iii) in Lemma \ref{lem: unstable certificate bs s} and Lemma \ref{lem: unstable randomized quantum bs} respectively. 

 \begin{algorithm}[h]
    \begin{algorithmic}[1]
        \State \textbf{Given:} Known $f : \zone^n \to \zone$; Query access to an unknown $x \in \zoneu^n$.
        \State \textbf{Goal:} Output $\uf(x)$
        \State Initialize partial assignment $x^*\gets *^n$
        \For{$i \gets 1$ to $\ubsone(\uf)$}\label{line:firstFor}
            \State{$c \gets$ a minimum 0-certificate of $\uf$ at an arbitrary $x \in \uf^{-1}(0)$ consistent with $x^*$}\label{line: firstCertificate}
            \State $C \gets$ domain of $c$
            \State Query all variables in $C$
            \State Update $x^*$ with the answers from the oracle.
            \If{ $x^*$ is a $0$-certificate of $\uf$}\label{line:firstIf}
                \State \textbf{Output} 0
            \EndIf
            \If{ $x^*$ is a $1$-certificate of $\uf$}\label{line:secondIf}
                \State \textbf{Output} 1
            \EndIf
        \EndFor

        \For{$i \gets 1$ to $\ubszero(\uf)$}\label{line:secondFor}
            \State{$c \gets$ a minimum 1-certificate of $\uf$ at an arbitrary $x \in \uf^{-1}(1)$ consistent with $x^*$}\label{line:secondCert}
            \State $C \gets$ domain of $c$
            \State Update $x^*$ with the answers from the oracle.
            \If{ $x^*$ is a $0$-certificate of $\uf$}\label{line:thirdIf}
                \State \textbf{Output} 0
            \EndIf
        \EndFor
        \State \textbf{Output $\su$}\label{line:lastOutput}

        \caption{\uquery algorithm}
        \label{algo:UpperBound}
    \end{algorithmic}
    \end{algorithm}

Recall that we call a string $y\in\zoneu^n$ \emph{consistent} with a partial assignment $x^*\in \{0,1,\su,*\}^n$ if for all $i\in [n]$, either $x^*_i=y_i$ or $x^*_i=*$. 
We call a string $y\in\zoneu^n$ \emph{inconsistent} with a partial assignment $x^*\in \{0,1,\su,*\}^n$ if there exists some index $i\in [n]$ such that $x_i\neq *$ and $x_i\neq y_i$.

\begin{claim}
    If Algorithm \ref{algo:UpperBound} reaches Line \ref{line:secondFor}, then every $1$-input of $\uf$ is inconsistent with the partial assignment $x^*$ (at Line \ref{line:secondFor}).
    \label{claim:firstInconsistency}
\end{claim}
\begin{proof}
Let $k$ denote $\ubsone(\uf)$, the number of iterations of the \textbf{for} loop on Line \ref{line:firstFor}. Assume the algorithm reached Line \ref{line:secondFor}, and let $x'$ be the partial assignment $x^*$ constructed by the algorithm when it reaches Line \ref{line:secondFor}. Suppose, for the sake of contradiction, there exists a input $y\in \uf^{-1}(1)$ that is consistent with $x'$.

The string $x'$ contains neither a $0$-certificate nor a $1$-certificate for $\uf$ (because otherwise one of the two \textbf{if} conditions on Lines \ref{line:firstIf} and \ref{line:secondIf} would have terminated the algorithm). Suppose the $0$-certificates used during the run of the \textbf{for} loop on Line \ref{line:firstFor} were $c_1, \ldots, c_k$, and their respective domains were $C_1,\ldots, C_k$. Since $x'$ is not a $0$-certificate, it must be the case that for each $i\in [k]$, the string $x'$ and $c_i$ differ in $C_i$. Let $B_i\subseteq C_i$ be the set of positions in $C_i$ where $x'$ differs from $c_i$. Observe that since $c_{i+1}$ is chosen to be consistent with $x^*$, it must be the case that $c_{i+1}$ and $x^*$ agree on all positions in $C_i$. Hence $B_{i+1}$ is disjoint from $B_{i}$. With the same reasoning, we can conclude that $B_{i+1}$ is disjoint from every $B_j$ where $j\le i$.

Observe that for each $i\in [k]$, there is a setting to the bits in $B_i$ such that $x'$ becomes a $0$-input --- simply take the setting of these bits from $c_i$, which is a 0-certificate. More formally, for each $i\in [k]$, there exists a string $\alpha_i\in \zone^{|B_i|}$ such that $f(x'|_{B_i\gets \alpha_i})=0$. 

Since $y$ is consistent with $x'$, it agrees with $x'$ on all positions where $x'\neq *$. This means the previous observation holds for $y$ too. That is, for each $i\in [k]$, there exists strings $\alpha_i\in \zone^{|B_i|}$ such that $f(y|_{B_i\gets \alpha_i})=0$. Recall that $y\in \uf^{-1}(1)$, and hence the sets $B_1,\ldots, B_k$ form a collection of disjoint sensitive blocks for $\uf$ at $y$. Further, since the algorithm has not found a $1$-certificate yet, it must be the case that there is some $0$-input consistent with $x'$. This means there is yet another disjoint block $B_{k+1}$ that is sensitive for $\uf$ at $y$. But this is a contradiction since the maximum, over all $1$-inputs of $\uf$, number of disjoint sensitive blocks is $\ubsone(\uf) = k < k+1$.
\end{proof}

\begin{claim}
If Algorithm \ref{algo:UpperBound} reaches Line \ref{line:lastOutput}, then the partial assignment $x^*$ (at Line \ref{line:lastOutput}) is inconsistent with every $0$-input, and every $1$-input of $\uf$.
    \label{claim:secondInconsistency}
\end{claim}
\begin{proof}
The fact that every $1$-input is inconsistent with $x^*$ follows from the previous claim. Showing that every $0$-input is inconsistent with $x^*$ can be done using a nearly identical proof as that of the last claim, and we omit it.
\end{proof}

\begin{theorem}\label{thm: algo upper bound}
    Algorithm~\ref{algo:UpperBound} correctly computes $\uf$, and makes at most $O(\usC(f)\ubs(f))$ queries. Thus $\sDu(f) = O(\usC(f) \cdot \ubs(f))$.
\end{theorem}
\begin{proof}
    If the algorithm outputs a value in $\zone$, then it must have passed the corresponding \textbf{if} condition (either in Line~\ref{line:firstIf}, or in Line~\ref{line:secondIf}, or in Line~\ref{line:thirdIf}), and is trivially correct. If the algorithm outputs $\su$, then it must have reached Line \ref{line:lastOutput}. From Claim \ref{claim:secondInconsistency}, we conclude that every $0$-input and every $1$-input of $\uf$ must be inconsistent with  the partial assignment $x^*$ constructed by the algorithm when it arrives at Line \ref{line:lastOutput}. This means that every $x \in \zoneu^n$ that is consistent with $x^*$ (in particular, the unknown input $x$) must satisfy $\uf(x) = \su$, which concludes the proof of correctness.

    The \textbf{for} loop on Line \ref{line:firstFor} runs for $\ubsone(f)= O(\ubs(f))$ many iterations, and at most $\usCzero(f)\leq \usC(f)$ many bits are queried in each iteration. So the first \textbf{for} loop makes a total of $O(\usC(f)\ubs(f))$ many queries. 
    Similarly, the number of queries by the second \textbf{for} loop on \ref{line:secondFor} is $O(\usC(f)\ubs(f))$, making the total number of bits queried by the algorithm as $O(\usC(f)\ubs(f))$.
\end{proof}

\begin{lemma}\label{lem: unstable certificate bs s}
    Let $f : \zone^n \to \zone$. Then $\usC(f) = O(\ubs(f) \cdot \us(f))$.
\end{lemma}
\begin{proof}
    Fix $x \in \zoneu^n$. Let $B_1, \dots, B_k$ be a maximal set of minimal disjoint sensitive blocks of $\uf$ at $x$. By the definition of $\ubs(f)$, we have $k \leq \ubs(f)$. We first show that for all $i \in [k]$, $|B_i| \leq \us(f)$. Towards this, fix $i \in [k]$. We have $b = \uf(x) \neq \uf(x^{B_i}) = b'$. Furthermore, for any $j \in B_i$, if we denote by $y$ the input obtained by flipping the value of the $j$'th coordinate of $x^{B_i}$ back to the value $x_j$, then we must have $\uf(y) = b$. This is because $B_i$ was assumed to be a \emph{minimal} sensitive block; if $\uf(y) \neq b$, then the set $B_i \setminus \cbra{j}$ would have been a smaller sensitive block for $\uf$ at $x$ than $B_i$. In particular, this means each $j \in B_i$ is sensitive for $x^{B_i}$.

    Next denote $S = \bigcup_{i = 1}^kB_i$. We claim that $x_S$ is a certificate for $\uf$ at $x$. Towards a contradiction, suppose $x_S$ is not a certificate for $\uf$ at $x$. This implies existence of a $y \in \zoneu^n$ such that $y_S = x_S$ and $\uf(y) \neq \uf(x)$. Let $B' = \cbra{i \in [n] : x_i \neq y_i}$. Thus $B'$ is a sensitive block for $\uf$ at $x$. Furthermore, $B' \cap S = \emptyset$, which means $B' \cap B_i = \emptyset$ for all $i \in [k]$. This contradicts our assumption that $B_1, \dots, B_k$ was a \emph{maximal} set of minimal disjoint sensitive blocks for $\uf$ at $x$.

    Combining the above, we conclude that $S = \bigcup_{i = 1}^kB_i$ is a certificate for $\uf$ at $x$, and its size is at most $k \cdot \us(f) \leq \ubs(f) \cdot \us(f)$, concluding the proof of the lemma.
\end{proof}

\begin{lemma}\label{lem: unstable randomized quantum bs}
    Let $f : \zone^n \to \zone$. Then,
    \[
    \sRu(f) = \Omega(\ubs(f)), \qquad \sQu(f) = \Omega(\sqrt{\ubs(f)}).
    \]
\end{lemma}
\begin{proof}
    We use Lemma~\ref{lem: yao} for the randomized lower bound. Let $x \in \zoneu^n$ be such that $\ubs(f) = \bs(\uf, x) = k$, say. Define a distribution $\mu$ on $\zoneu^n$ as follows:
    \begin{align*}
        \mu(x) & = 1/2,\\
        \mu(x^{B_i}) & = 1/2k~\textnormal{for all}~i \in [k].
    \end{align*}
    Towards a contradiction, let $T$ be a deterministic decision tree of cost less than $k/10$ that computes $\uf$ to error at most $1/3$ under the input distribution $\mu$. Let $L_x$ denote the leaf of $T$ reached by the input $x$. There are now two cases:
    \begin{itemize}
        \item If the output at $L_x$ is not equal to $\uf(x)$, then $T$ errs on $x$, which contributes to an error of 1/2 under $\mu$, which is a contradiction.
        \item If the output at $L_x$ equals $\uf(x)$, since the number of queries on this path is less than $k/10$, there must exist at least $9k/10$ blocks $B_i$ from the set $B_1, \dots, B_k$ such that no variable of $x^{B_i}$ is read on input $x^{B_i}$ (observe that in this case, $x^{B_i}$ also reaches the leaf $L_x$). Since $f(x^{B_i}) \neq f(x)$, this means $T$ makes an error on each of these $B_i$'s, contributing to a total error of at least $9k/10 \cdot 1/2k = 0.45$ under $\mu$, which is a contradiction.
    \end{itemize}
    This concludes the randomized lower bound.

    For the quantum lower bound we use Theorem~\ref{thm: ambainis}. Define $X = \cbra{x}, Y = \cbra{x^{B_i} : i \in [k]}$, and $R = X \times Y$. From Theorem~\ref{thm: ambainis} we have $m = k, m' = 1, \ell = 1$ (since each index appears in at most 1 block, as each block is disjoint) and $\ell' = 1$. Theorem~\ref{thm: ambainis} then implies $\sQu(f) = \Omega(\sqrt{k}) = \Omega(\sqrt{\ubs(f)})$.
\end{proof}

We conclude this section by proving Theorem~\ref{thm: main}.
\begin{proof}[Proof of Theorem~\ref{thm: main}]
    Theorem~\ref{thm: algo upper bound} and Lemma~\ref{lem: unstable certificate bs s} imply
    \[
        \sDu(f) = O(\usC(f) \cdot \ubs(f)) = O(\ubs(f)^2 \cdot \us(f)) = O(\ubs(f)^3) = O(\sRu(f)^3),
    \]
    where the final bound is from Lemma~\ref{lem: unstable randomized quantum bs}. Substituting the second bound from Lemma~\ref{lem: unstable randomized quantum bs} in the last equality above also yields $\sDu(f) = O(\sQu(f)^6)$.
\end{proof}

\section{Conclusions and future work}
In this paper we initiated a study of $\su$-query complexity of Boolean functions, by considering a natural generalization of Kleene's strong logic of indeterminacy on three variables. On the one hand, we showed that the Indexing function which admits a logarithmic-cost query protocol in the usual query model, has linear query complexity in the $\su$-query model. 
On the other hand, we showed that a monotone variant of it admits an efficient \uquery protocol. More generally, we showed that for monotone Boolean functions, the $\su$-query complexity and usual query complexity are within a factor of two of each other.

We showed using $\su$-analogs of block sensitivity, sensitivity, and certificate complexity that for all Boolean functions $f$, its deterministic, randomized, and quantum $\su$-query complexities are polynomially related to each other. An interesting research direction would be to determine the tightest possible separations between all of these measures. It is interesting to note that our cubic relationship between deterministic and randomized $\su$-query complexity matches the best-known cubic relationship in the usual query models. However, it is known that $\sD(f) = O(\sQ(f)^4)$~\cite{ABKRT21}, which is essentially the best separation possible in the usual query model~\cite{ABBLSS17}. This proof goes via giving a lower bound on quantum query complexity of a Boolean function by the intermediate measure of its \emph{degree}.
It is interesting to see if the techniques of~\cite{ABBLSS17} can be adapted to show a fourth-power relationship between deterministic and quantum $\su$-query complexities in our setting as well.
It would also be interesting to see if $\su$-query complexity of Boolean functions can be characterized by more standard measures of complexity of Boolean functions. More generally, it would be interesting to see best known relationships between combinatorial measures of Boolean functions in this model, and see how they compare to the usual query model (see, for instance,~\cite[Table~1]{ABKRT21}).

While we studied an important extension of Boolean functions to a specific three-valued logic that has been extensively studied in various contexts, an interesting future direction is to consider query complexities of Boolean functions on other interesting logics. Our definition of $\uf$ dictates that $\uf(x) = b \in \zone$ iff $f(y) = b$ for all $y \in \res(x)$, and $\uf(x) = \su$ otherwise. A natural variant is to define a 0-1 valued function that outputs $b \in \zone$ iff majority of $f(y)$ equal $b$ over all $y \in \res(x)$. The same proof technique we used in Lemma~\ref{lem: unstable indexing lower bounds} can easily be extended to show that the complexity of this variant of $\IND_n$ is bounded from below by the usual query complexity of Majority on $2^n$ variables, which is $\Omega(2^n)$ in the deterministic, randomized, and quantum query models.

\section*{Acknowledgements} We thank anonymous reviewers for several useful comments.

\bibliography{bibo}

\end{document}